\documentclass[reprint,amsmath,amssymb,aps,prl]{revtex4-2}

\usepackage{graphicx}      
\usepackage{hyperref} 
\usepackage{xcolor}        
\usepackage{mathrsfs}      
\usepackage{physics}
\newcommand{\h}{\mathfrak{h}}
\newcommand{\g}{\mathfrak{g}}
\newcommand{\E}{\mathbb{E}}

\def\Xint#1{\mathchoice
   {\XXint\displaystyle\textstyle{#1}}%
   {\XXint\textstyle\scriptstyle{#1}}%
   {\XXint\scriptstyle\scriptscriptstyle{#1}}%
   {\XXint\scriptscriptstyle\scriptscriptstyle{#1}}%
   \!\int}
\def\XXint#1#2#3{{\setbox0=\hbox{$#1{#2#3}{\int}$}
     \vcenter{\hbox{$#2#3$}}\kern-.5\wd0}}

\def\dashint{\Xint-}

\usepackage{amsthm} 

\newtheorem{lemma}{Lemma}

\begin{document}

\title{Free Convolution and Generalized Dyson Brownian Motion}

\author{Pierre Bousseyroux}
\email{pierre.bousseyroux@polytechnique.edu}
\affiliation{Chair of Econophysics and Complex Systems, Ecole polytechnique, 91128 Palaiseau Cedex, France}
\affiliation{LadHyX UMR CNRS 7646, Ecole polytechnique, 91128 Palaiseau Cedex, France}

\author{Jean-Philippe Bouchaud}
\email{jean-philippe.bouchaud@cfm.com}
\affiliation{Capital Fund Management, Paris 75007, France}
\affiliation{Chair of Econophysics and Complex Systems, Ecole polytechnique, 91128 Palaiseau Cedex, France}
\affiliation{Academie des Sciences, Paris 75006, France}
\begin{abstract}
The eigenvalue spectrum of the sum of large random matrices that are mutually ``free'', i.e., randomly rotated, can be obtained using the formalism of R-transforms, with many applications in different fields. We provide a direct interpretation of the otherwise abstract additivity property of R-transforms for the sum in terms of a dynamical evolution of ``particles'' (the eigenvalues), interacting through two-body and higher-body forces and subject to a Gaussian noise, generalizing the usual Dyson Brownian motion with Coulomb interaction. Interestingly, the appearance of an outlier outside of the bulk of the spectrum is signalled by a divergence of the ``velocity'' of the generalized  Dyson motion. We extend our result to products of free matrices. 
\end{abstract}

\maketitle

\vskip 1cm
One of the most exciting recent developments in Random Matrix Theory is the concept of ``freeness'', which generalizes independence for non-commuting random objects, like matrices \cite{voiculescu1992free, speicher2009free, tulino2004random}. In a nutshell, two large symmetric random matrices $\mathbf{A}$ and $\mathbf{B}$ are said to be free if their eigenbasis are ``as different'' as possible. More formally, if one chooses a basis such a $\mathbf{A}$ is diagonal, then $\mathbf{B}$ is free if it can be written as $\vb{O} D \vb{O}^\top$, where $D$ is a diagonal matrix and $\vb{O}$ a random orthogonal matrix, chosen uniformly over the rotation group. In that sense, free matrices are maximally non-commuting, whereas $\mathbf{A}$ and $\mathbf{B}$ would instead commute if their eigenbasis coincide (up to a permutation).

In general, the eigenvalues of the sum of two non-commuting matrices -- say $\mathbf{A}$ and $\mathbf{B}$ -- cannot be simply characterized in terms of the eigenvalues of each of these matrices. But when one considers free symmetric matrices of large sizes, the distribution of the eigenvalues of the sum can be exactly computed when the distribution of eigenvalues of $\mathbf{A}$ and $\mathbf{B}$ are known. More precisely, one relies on a generalisation of the log-characteristic function for classical random variables called the R-function in the context of free matrices, which is additive when $\mathbf{A}$ and $\mathbf{B}$ are free, i.e. $R_{\mathbf{A}+\mathbf{B}}(z) = R_{\mathbf{A}}(z)+R_{\mathbf{B}}(z)$, where the coefficient of the expansion of $R_{\mathbf{A}}(z)$ (resp. $R_{\mathbf{B}}(z)$) in powers of $z$ define the {\it free cumulants} $\kappa_n$ of the matrix $\mathbf{A}$ (resp.  $\mathbf{B}$) -- see e.g. \cite{tulino2004random, speicher2009free, potters2020first} and below. In other words, free cumulants are additive, exactly as usual cumulants are additive when one deals with sums of independent random variables.  Similar results hold for products of free matrices, i.e. the eigenvalue spectrum of $\sqrt{\mathbf{B}}\, \mathbf{A}\sqrt{\mathbf{B}}$ can be obtained from those of $\mathbf{A}$ and $\mathbf{B}$. These properties have proven to be extremely useful in a host of different fields, from quantum chaos, telecommunication, machine learning, ecology and finance, see e.g. \cite{Foini2019ETH,Foini2019rotinv,Pappalardi2022,tulino2004random, stone2018feasibility, baron2022eigenvalue,yan2017edge, bouchbinder2021low, potters2020first, atanasov2024scaling} for various examples and reviews. 

The aim of this work is to interpret the additivity property of R-transforms in physical terms through effective dynamical equations at the level of individual eigenvalues. Such an interpretation is well-known in the case where  $\mathbf{B}$ is a Wigner matrix, i.e. a symmetric $N \times N$ matrix where all entries $B_{ij}$ are IID random variables with zero mean and a finite variance $\sigma^2/N$, such that only the second free cumulant $\kappa_2^\mathbf{B}$ is non zero. In this case, the eigenvalues of $\mathbf{A}+\mathbf{B}$ can be obtained as the solutions of the following ``Dyson Brownian motion'' equations computed at fictitious time $t=1$:
\begin{equation}\label{eq:DBM}
    \mathrm{d}\lambda_i = \sqrt{\frac{2\sigma^2}{N}}\mathrm{d}W_i + \frac{\sigma^2}{N} \sum_{1\leq j\neq i\leq N} \frac{ \mathrm{d}t}{\lambda_i - \lambda_j},\quad (i=1, \ldots, N),
\end{equation}
where $\lambda_i(t=0)=\lambda_i^\mathbf{A}$ are the eigenvalues of $\mathbf{A}$ and $\mathrm{d}W_i$ are independent normalized Wiener noises. 

In the following, we generalize these equations of motion for matrices $\mathbf{B}$ with arbitrary free cumulants $\kappa_n^\mathbf{B}$. We first consider the deterministic case and consider a continuous sequence of random matrices $t\mapsto \vb{M}(t)= \vb{A} + \vb{B}(t)$ with $\vb B(1)=\vb B$. The corresponding random eigenvalues are $\lambda_1(t) < \ldots < \lambda_N(t)$. We then prove that this leads to a deterministic, effective evolution equations of these eigenvalues which reads, in the $N \gg 1$ limit
\begin{equation}\label{eqmain}
    \mathrm{d}\lambda_i =  F_i(\lambda_i(t),t) \mathrm{d}t + o(1), \qquad (i=1, \ldots, N),
\end{equation}
with
\vspace{-0.5em}
\begin{equation}\label{eq:GDBM}
    F_i(\lambda_i,t) = \sum_{n=1}^{+\infty}  \, \frac{\mathrm{d} \kappa_{n}^{\vb{B}(t)}}{\mathrm{d}t}\frac{1}{N^{n-1}} \sum_{\substack{e_1,\dots,e_{n-1} \neq i \\ \text{p.d.}}} \frac{1}{\prod_{j=1}^{n-1} (\lambda_i - \lambda_{e_j})},
\end{equation}
where the sum runs over all indices $1 \leq e_1, \dots, e_{n-1} \leq N$ that are pairwise distinct (``p.d.'') and different from $i$. 

Note that equation of motion Eq. \eqref{eqmain} involves n-body ``forces'' with $n \geq 2$, whereas the standard case Eq.~\eqref{eq:DBM} only contains two body forces, identical to the term $n=2$ in Eq.~\eqref{eq:GDBM} with $\kappa_2^{\vb{B}}=\sigma^2$. Furthermore, these n-body forces are non-potential, i.e. cannot be written as $-\partial_i V(\{ \lambda_j \})$ where $V$ is a function of the position of all eigenvalues, except the two-body force which corresponds to the well-known Coulomb gas $V(\{ \lambda_j \}) = - \sum_{i < j} \log |\lambda_i - \lambda_j |$ \cite{forrester2010log}. Non-potential effects, however, are of order $O(N^{-1})$. In fact, when $N \to \infty$, $F_i(\lambda_i,t)$ converges to $\sum_{n=1}^{+\infty}  \, \dot \kappa_{n}^{\vb{B}(t)} \h_{n-1}(\lambda_i,t)$, where we have introduced generalized Hilbert transforms of the eigenvalues, defined as 
\begin{equation}\label{defh}
\h_{n}(\lambda,t) = \dashint \dots \dashint \prod_{k=1}^n {\rm d}\lambda_k \rho(\lambda_k,t) \frac{1}{\prod_{k=1}^{n} (\lambda - \lambda_k)}, 
\end{equation}
where $\rho(\lambda,t)$ is the time dependent density of eigenvalues, and $\dashint$ means that we are taking the principal value of the integrals. We note that $\h_0(\lambda)=1$ and $\h(\lambda):=\h_1(\lambda)$ is the standard Hilbert transform. Using Matytsin's ``trick'' \cite{matytsin1994large}, one can in principle compute all $\h_{n >1}(\lambda)$ in terms of $\h(\lambda)$ and the eigenvalue density $\rho$. For example, $\h_2(\lambda)=\h^2(\lambda) - 2 \zeta(2) \rho^2(\lambda)$.

Let us first show how the effective dynamical equation of individual eigenvalues, Eq. \eqref{eq:GDBM}, recovers the standard additivity of R-transforms for free matrices. We will then show that  Eq.~\eqref{eq:GDBM} can also be obtained from standard perturbation theory from Quantum Mechanics \cite{landau2013quantum}, that can be resummed to all orders when the perturbation is free. Such an approach will also allow us to obtain the random contribution generalizing the Brownian term in Eq.~\eqref{eq:DBM} in the case where $\vb{B}$ is infinitely divisible in a free sense \cite{e2005classical}. 

The basic object that is needed is the Stieltjes transform of the density of eigenvalues (also called the normalized trace of the resolvent), defined as
\begin{equation}
    \g(z) := \frac{1}{N} \sum_{j = 1}^N \frac{1}{z - \lambda_j}
\end{equation}
for an arbitrary matrix with eigenvalues $\lambda_j$, and with $z \in \mathbb{C}$. The R-transform is then obtained as \cite{tulino2004random, speicher2009free, potters2020first}
\begin{equation}\label{defR}
   R(\g(z)):= z - \frac{1}{\g(z)},
\end{equation}
such that its expansion in powers of $g$ defines the free cumulants $\kappa_n$ such that $R(g) = \sum_{n=0}^{+\infty} \kappa_{n+1} g^{n}$.

Now, in order to prove our main result Eq.~\eqref{eqmain}, we will first obtain a generalization of the Burgers equation that describes the evolution of the Stieltjes transform of $\vb{M}(t) = \vb{A} + \vb{B}(t)$ where $\vb{B}$ is a Wigner matrix. We denote as  $\g_{\vb{M}}(z,t)$ the Stieltjes transform of $\vb{M}(t)$ evaluated at $z$, and similarly for R-transforms. Then, from that additivity of R-transforms (i.e. $R_{\vb{M}(t)}=R_{\vb{A}}+R_{\vb{B}(t)}$) when $\vb{A}$ and $\vb{B}(t)$ are free and using the methods of characteristics, one obtains (see proof in \hyperref[proofburgers]{SI-1}): 
\begin{equation}\label{burgers}
    \partial_t \g_{\vb{M}}(z,t) = - \partial_t R_{\vb{B}}(\g_{\vb{M}}(z,t),t) \partial_z \g_{\vb{M}}(z,t).
\end{equation}
Note that when $\vb{B}(t)$ is a Wigner matrix of variance $\sigma^2 t $, one has $R_{\vb{B}}(g,t) = \sigma^2 t g$ and one recovers the well-known Burgers equation for $\g_{\vb{M}}(z,t)$, namely (see e.g. \cite{rogers1993interacting, allez2012invariant, potters2020first}) 
\begin{equation}\label{burgers_W}
    \partial_t \g_{\vb{M}}(z,t) = - \sigma^2 \g_{\vb{M}}(z,t) \partial_z \g_{\vb{M}}(z,t).
\end{equation}
Now since 
\begin{equation}
    \partial_z \g_{\vb{M}} = -\frac{1}{N} \sum_{e_0 = 1}^{N} \frac{1}{(z - \lambda_{e_0})^2}, 
\end{equation}
where $\lambda_{e_0}$ are the eigenvalues of $\vb{M}(t)$, one can write the right hand side of Eq.~\eqref{burgers} as
\begin{align} \nonumber
    -\partial_t R_{\vb{B}}(\g_{\vb{M}},t) \partial_z \g_{\vb{M}} 
    &= \sum_{n=0}^{+\infty} \frac{\mathrm{d} \kappa_{n+1}^{\vb{B}(t)}}{\mathrm{d}t} n(\g_{\vb{M}})^{n-1}
    \frac{1}{N} \sum_{e_0 = 1}^{N} \frac{1}{(z - \lambda_{e_0})^2}
\end{align}
\vspace{-5mm} 
\begin{align}
    &= \sum_{n=1}^{+\infty} \frac{\mathrm{d} \kappa_{n+1}^{\vb{B}(t)}}{\mathrm{d}t} \frac{n}{N^{n}} 
    \sum_{\substack{e_0, \ldots, e_{n-1} \\ \text{p.d.}}} 
    \frac{1}{(z - \lambda_{e_0})^2} 
    \frac{1}{\prod_{k=1}^{n-1} (z - \lambda_{e_k})}
\end{align}
where we have assumed that when $N \to \infty$, all $e_i$'s are distinct. The key now is to use the Lemma \ref{lemma1}, shown in \hyperref[s2]{SI-2}, 
\begin{align}\label{eq1}
    \partial_t \g_{\vb{M}}(z,t) &= -\partial_t R_{\vb{B}}(\g_{\vb{M}},t) \partial_z \g_{\vb{M}} \nonumber \\
    &= \frac{1}{N} \sum_{e_0} \frac{1}{(z - \lambda_{e_0})^2} F(\lambda_{e_0},t).
\end{align}
where $F(\lambda_{e_0},t)$ is given by Eq.~\eqref{eq:GDBM}.
Hence, introducing the time dependent density of eigenvalues $\rho(\lambda,t)$, and dropping unnecessary indices, one can write:
\begin{equation}
\partial_t g =  \int \frac{\partial_t \rho(\lambda,t)}{z - \lambda} \, {\rm d}\lambda = -\int_{\lambda} \frac{\partial_{\lambda}(\rho(\lambda,t) F(\lambda,t))}{z - \lambda} \, 
{\rm d}\lambda,
\end{equation}
where we have used Eq.~\eqref{eq1} and integrated by parts the last term. Since this equality is true for all $z$, we deduce that the time evolution of the eigenvalue density is given, for $N \to \infty$, by the following Euler equation
\begin{equation}
    \partial_t \rho(\lambda,t) = - \partial_x \left[\rho(\lambda,t) F(\lambda,t)\right],
\end{equation}
which can be interpreted as the time evolution of the density of ``particles'' the position of which evolve according to $\dot \lambda=F(\lambda,t)$. This proves our central result, Eq.~\eqref{eqmain}. 
As an additional check, in the case where $\vb{B}(t)={t}^\alpha\,\vb{B}$ with $\alpha > 0$, one can take the $t \to \infty$ limit of Eq.~\eqref{eqmain} where the eigenvalues $\lambda_i$ should behave asymptotically as ${t}^\alpha\lambda_i^{\vb{B}}$. Substituting into Eq.~\eqref{eqmain}, we find that the following equation should be satisfied when $ s_i = \lambda_i^{\vb{B}} $, since $ \kappa_{n}^{\vb{B}(t)} = t^{\alpha n} \kappa_{n}^{\vb{B}} $:
\begin{equation}\label{moyen}
    s_i \underset{N \to \infty}{=} \sum_{n=1}^{\infty} 
    (n+1) \kappa_{n+1}^{\vb{B}} \,
    \h_{n}^{(N)}(s_i) + o(1),
\end{equation}
where $\h_{n}^{(N)}(s_i)$ are the finite $N$ generalized Hilbert transforms, i.e.
\begin{equation}
    \h_{n}^{(N)}(s_i) := \frac{1}{N^n} \sum_{\substack{e_1,\dots,e_n \neq i \\ \text{p.d.}}} \frac{1}{\prod_{j=1}^{n} (s_i - s_{e_j})}
\end{equation}
Multiplying both sides by $(z - s_i)^{-1}$, summing over $i$ and setting $s_i=\lambda_i^{\vb{B}}$ leads, after a few manipulations, to $ z \g_{\vb{B}}(z) - 1 = \g_{\vb{B}}(z)R_{\vb{B}}(\g_{\vb{B}}(z))$, which indeed coincides with the definition of the R-transform in terms of the Stieltjes transform $\g(z)$, Eq.~\eqref{defR}.

Let us now give an interpretation of the (fictitious) dynamical evolution of the eigenvalues, Eq.~\eqref{eq:GDBM} from the point of view of perturbation theory, treating $t \,\vb{B}$ as small when $t \to 0$. The eigenvalues of $\vb{M}(t) = \vb{A} + t\vb{B}$ can then be expanded as \cite{landau2013quantum}: 
\begin{equation}\label{eq:pert}
    \lambda_i \underset{t \to 0}{=} \lambda_i^{\vb{A}} + t \mu_i^{(1)} + t^2 \mu_i^{(2)} + \dots + t^{k} \mu_i^{(k)} + o(t^{k}),
\end{equation}
with 
\begin{equation}\label{per1}
    \mu_i^{(1)} = \bra{i} \vb{B} \ket{i}, \quad  \mu_i^{(2)} = \sum_{k \neq i} \frac{\bra{k} \vb{B} \ket{i}^2}{\lambda_i^{\vb{A}} - \lambda_k^{\vb{A}}},
\end{equation}
\begin{multline}\label{per2}
    \mu_i^{(3)} = \sum_{k_1, k_2\neq i}\frac{\bra{i} \vb{B} \ket{k_1} \bra{k_1} \vb{B} \ket{k_2} \bra{k_2} \vb{B} \ket{i}}{(\lambda_i^{\vb{A}} - \lambda_{k_1}^{\vb{A}})(\lambda_i^{\vb{A}} - \lambda_{k_2}^{\vb{A}})} \\- \bra{i} \vb{B} \ket{i} \sum_{1 \leq k \neq i \leq N} \frac{|\bra{i} \vb{B} \ket{k}|^2}{(\lambda_i^{\vb{A}} - \lambda_k^{\vb{A}})^2},
\end{multline}
and more and more complicated expressions as $k$ increases ($\mu_i^{(4)}$ is given in the SI-3, Eq.~\eqref{mu4}). The bra-kets $\bra{a}$, $\ket{a}$ denote the unit eigenvectors of $\vb{A}$, respectively associated with the eigenvalues $\lambda_a^{\vb{A}}$. More generally, as shown in the third point of the SI, $\mu_i^{(j+1)}$ includes only sums of fractional terms where the numerator is the product of $ j $ matrix elements of $ \vb{B} $, and the denominator is of the form $ (\lambda_i^{\vb{A}} - \lambda_{k_1}^{\vb{A}}) \dots (\lambda_i^{\vb{A}} - \lambda_{k_j}^{\vb{A}}) $, where $ \lambda_{k_1}^{\vb{A}}, \dots, \lambda_{k_j}^{\vb{A}} $ are eigenvalues of $ \vb{A} $ different from $ \lambda_i^{\vb{A}} $. The contribution from the case where the $ \lambda_{k_a}^{\vb{A}} $ are distinct is of the form
    \begin{equation}\label{firsteq}
        \sum_{\substack{k_1, \dots, k_j \\ \text{\text{p.d.}}}} \frac{\bra{i} \vb{B} \ket{k_1} \bra{k_1} \vb{B} \ket{k_2} \dots \bra{k_j} \vb{B} \ket{i}}{(\lambda_i^{\vb{A}} - \lambda_{k_1}^{\vb{A}}) \dots (\lambda_i^{\vb{A}} - \lambda_{k_j}^{\vb{A}})}.
    \end{equation}
Now, to leading order, all these quantities converge for large $N$ to their average over the random rotation defining the eigenbasis of $\vb{B}$, for which we find $ \E\left[\mu_i^{(1)}\right] = 0$ and
\begin{equation}\label{expected}
\begin{aligned}
\E\left[\mu_i^{(2)}\right] &\underset{N\to \infty}{=}\frac{\kappa_2^{\vb{B}}}{N} \sum_{1 \leq k \neq i \leq N} \frac{1}{\lambda_i^{\vb{A}} - \lambda_k^{\vb{A}}}, \\
\smash{\E\left[\mu_i^{(3)}\right]} &\underset{N\to \infty}{=}\frac{\kappa_3(\vb{B})}{N^2} 
\sum_{\substack{k_1, k_2\neq i \\ \text{p.d.}}}
\frac{1}{(\lambda_i^{\vb{A}} - \lambda_{k_1}^{\vb{A}})(\lambda_i^{\vb{A}} - \lambda_{k_2}^{\vb{A}})}.
\end{aligned}
\end{equation}
and more generally, the contribution to $\E\left[\mu_i^{(j+1)}\right]$ from the case where the $ \lambda_{k_a}^{\vb{A}} $ are distinct is
    \begin{equation}\label{secondeq}
        \frac{\kappa_{j+1}^{\vb{B}}}{N^j} \sum_{\substack{k_1, \dots, k_j \\ \text{\text{p.d.}}}} \frac{1}{(\lambda_i^{\vb{A}} - \lambda_{k_1}^{\vb{A}}) \dots (\lambda_i^{\vb{A}} - \lambda_{k_j}^{\vb{A}})},
    \end{equation}
which follows from a beautifully simple formula for the free cyclic cumulants obtained in \cite{maillard2019high} (see also \cite{bernard2024structured}):
\begin{equation}\label{cyclicformula}
    \E\left[\bra{b_1} \vb{B} \ket{b_2} \bra{b_2} \vb{B} \ket{b_3} \dots \bra{b_{j-1}} \vb{B} \ket{b_{j}}\right] \underset{N\to \infty}{\sim} \frac{\kappa_j^{\vb{B}}}{N^{j-1}}
\end{equation}
when the integers $1\leq b_1, ..., b_j\leq N$ are pairwise distinct and $\vb{B}$ is a rotationally invariant random matrix.

As we show in SI-4, the perturbation expansion in Eq.~\eqref{eq:pert} -- which is expected to converge to its mean value in the high-dimensional limit -- is the exact solution of our dynamical system in Eq.~\eqref{eqmain} for $t \to 0$, up to order $t^4$, i.e., including the contribution $\E\left[\mu_i^{(4)}\right]$. Furthermore, one can convince oneself that there is indeed a term corresponding to Eq.~\eqref{secondeq} when expanding the solution of Eq.~\eqref{eqmain} to order $t^{(j+1)}$. Since the additivity of R-transforms can be proven by re-summing the expansion of the Stieltjes transform of $\vb{A} + \vb{B}$ in powers of $\vb{B}$ (see e.g. \cite{potters2020first, atanasov2024scaling} and SI) and that perturbation theory is also result of such an expansion, we expect that the solution of Eq.~\eqref{eqmain} indeed  reproduces perturbation theory to all orders in $t$. 

We now turn to the non-deterministic case and aim to generalize the Dyson Brownian motion, Eq.~\eqref{eq:DBM}, by focussing on infinitely divisible matrices in a free sense \cite{e2005classical}. Infinite divisibility  means that one can write $\vb{B}(t) = \int _0^t{\rm d}\vb{B}(t)$, where ${\rm d}\vb{B}(t)$ is the sum of an infinitesimal Wigner matrix ${\rm d}\vb{W}(t)$ and an independent rank-one Poisson noise, more precisely:
\begin{equation}\label{poisson}
    {\rm d}\vb{B}(t) = {\rm d}\vb{W}(t) + \begin{cases}
        \eta_t \left| \vec U_t \right\rangle \left\langle \vec U_t \right| \quad & \text{prob.} \quad N {\rm d}t \\
        0 \quad & \text{prob.} \quad 1 - N {\rm d}t,
    \end{cases}   
\end{equation}
where ${\rm d}\vb{W}(t)$ is a (matrix) Wiener noise, $\eta_t$ are IID random variables with some distribution $\psi(\eta)$ and $\vec U_t$ are random vectors chosen uniformly on the unit sphere, $(\vec U_t)^2=1$. It is easy to show that the resulting second free cumulant of $\vb B(t)$ reads: $\kappa_2^{\vb B(t)}=(\sigma_W^2 + \mathbb{E}[\eta^2]) t$ (where $\sigma_W^2$ is the Wigner contribution), and higher cumulants are given by $\kappa_n^{\vb B(t)}=\mathbb{E}[\eta^n] t$. 

In Eq.~\eqref{eq:pert} we applied perturbation theory to $\vb{A}+t\vb{B}$ by considering $t\vb{B}$ as a small, independent perturbation of $\vb{A}$ as $t\to 0$, which allowed us to compute the initial derivatives of $\lambda_i(t)$ at $t=0$. In contrast, we now fix a time $t$, consider a small time step $\Delta t$, and write $\vb{M}(t + \Delta t) = \vb{M}(t) + \Delta \vb{B}(t)$, 
where $\Delta \vb{B}(t)$, defined as $\Delta \vb{B}(t) = \int_{t}^{t+\Delta t} \mathrm{d} \vb{B}(t)$, is treated as an independent matrix increment of $\vb{M}(t)$. We will then take the limit $\Delta t \to 0$ in the perturbation theory formula and show that contributions from all orders -- i.e. not only the first order -- contribute to the derivative of $\lambda_i(t)$ at time $t$. Accordingly, the general formula from SI-3 yields the eigenvalue shift for $t\to t+\Delta t$ as
\begin{align}\label{eq:eigen_perturbation}
  \Delta\lambda_i(t) &= \sum_{n=1}^\infty \sum_{1 \leq e_1, \dots, e_n \leq N} \frac{1}{2\pi i} \\ \nonumber &\oint dz \frac{z \, \bra{e_1} \Delta \vb{B} \ket{e_2}\bra{e_2} \Delta \vb{B} \ket{e_3} \dots \bra{e_n} \Delta \vb{B} \ket{e_1}}{(z - \lambda_{e_1})^2(z - \lambda_{e_2})\dots(z - \lambda_{e_n})}.
\end{align} We first assume $\sigma_W^2 = 0$, which leads to $\kappa_n^{\Delta \vb{B}} = \E[\eta^n]\Delta t$. It is shown in \cite{maillard2019high} (see also \cite{bernard2024structured}) that when the indices $e_i$ are not pairwise distinct, the expected value of the product of matrix elements in Eq.~\eqref{eq:eigen_perturbation} involves free cumulants of $\Delta \vb{B}$ (see also \cite{Foini2019ETH, Foini2019rotinv, Pappalardi2022, Bouverot2024}). Since free cumulants of $\Delta \vb{B}$ are all proportional to $\Delta t$, when considering the expectation and subsequently taking the limit $N \to \infty$, the only surviving terms linear in $\Delta t$ are those for which the indices $e_i$ are pairwise distinct. We thus obtain
\begin{align}\label{eq:expected_delta_lambda}
  \mathbb{E}[\Delta\lambda_i(t)] &\underset{\Delta t \to 0}{\sim} \Delta t \sum_{n=1}^\infty \mathbb{E}[\eta^n] \h_{n-1}^{(N)}(\lambda_i(t)),
\end{align}
which exactly recovers Eq.~\eqref{eq:GDBM}, as it should. 

Consider now the specific case of the largest eigenvalue, denoted $\lambda_{\mathrm{max}} = \lambda_N$. Letting $\rho_t(\lambda)$ be the limiting spectral density of eigenvalues of $\vb{M}(t)$ and assuming $\rho_t(\lambda_{\mathrm{max}}) = 0$, the condition of pairwise distinctness for indices $e_i$ becomes irrelevant in the definition of $\h_n$ (see Eq.~\eqref{eq:GDBM}). In this case one obtains $\h_n(\lambda) = \h^n(\lambda)$, where we recall that $\h$ is the standard Hilbert transform of $\rho$. Since $\rho_t=0$ at the edge of the spectrum, one also has $\g_t(\lambda_{\mathrm{max}})=\h_{t}(\lambda_{\mathrm{max}})$. Hence, Eq.~\eqref{eq:expected_delta_lambda} simplifies to
\begin{equation} \label{eq:BBP}
    {\rm d} \lambda_{\mathrm{max}} = \mathbb{E}\left[\frac{\eta_t}{1 - \eta_t \g_{t}(\lambda_{\mathrm{max}})}\right] {\rm dt}.
\end{equation}
This expression is valid only if $\eta_t \g_{t}(\lambda_{\mathrm{max}}(t)) < 1$, which coincides exactly with the condition for the absence of an outlier when adding a rank-1 perturbation $\eta_t \ket{U_t}\bra{U_t}$ to matrix $\vb{M}(t)$ (see \cite{baik2005phase,benaych2011eigenvalues, potters2020first}). Quite interestingly, the appearance of an outlier beyond $\lambda_{\mathrm{max}}$ is signalled by a divergence of the ``velocity'' in our generalized  Dyson framework.

Returning to the analysis of a general eigenvalue $\lambda_i(t)$, we can examine the fluctuations of $\Delta \lambda_i(t)$ around its expected value  $\mathbb{E}[\Delta\lambda_i(t)]$. Again, the reasoning above shows that only terms with pairwise distinct indices contribute significantly. Furthermore, Lemma \ref{lemma} in SI-5 shows that the variance of $\bra{e_1}\Delta \vb{B}\ket{e_2}\bra{e_2}\Delta \vb{B}\ket{e_3}\dots\bra{e_n}\Delta \vb{B}\ket{e_1}$ scales as $(\Delta t)^j$ for some $j\geq 2$ as $N\to\infty$. Thus, in the subsequent limit $\Delta t\to 0$, the random fluctuations of $\Delta\lambda_i$ arise solely from the leading-order term corresponding to $j=1$, which is simply $\Delta \vb{B}_{ii}:=\bra{i}\Delta\vb{B}\ket{i}$.

Using the results of \cite{guionnet2005fourier}, one can show that in the high-dimensional limit the classical cumulants $c_n$ of any diagonal element of a rotationally invariant matrix (such as $\Delta \vb{B}$) are related to the free cumulants $\kappa_n$ of the full matrix via
\begin{equation}
c_n^{\Delta \vb{B}_{ii}} \underset{N\to\infty}{\sim} \left(\frac{2}{N}\right)^{n-1}(n-1)!\,\kappa_n^{\Delta \vb{B}}.
\end{equation}

Hence to leading order in $N$, we obtain the following Generalized Dyson Brownian Motion (GDBM) describing the evolution of eigenvalues, as a set of interacting particles with a Brownian noise:
\begin{align}\label{eq:final_DBM} 
    \mathrm{d}\lambda_i(t) =& \sum_{n=1}^{\infty} \mathbb{E}[\eta^n]\, \h_{n-1}^{(N)}(\lambda_i(t))\,\mathrm{d}t \\ \nonumber
     &+ \sqrt{\frac{2 \sigma^2}{N}}\,\mathrm{d}W_i(t) + O(N^{-1}), \,\, i = 1,\dots,N
\end{align}
where $\sigma^2= \sigma^2_W + \mathbb{E}[\eta^2]$ and $\mathrm{d}W_i(t)$ denotes independent normalized Wiener noises.

Finally, let us mention that a dynamical system analogous to Eq.~\eqref{eq:GDBM} can be constructed to interpret the {\it product} of free matrices and the multiplicativity of ``S-transforms''. In this case, the ``velocity'' $F_i(\lambda_i,t)$ corresponding to the matrix $\sqrt{\vb{B}(t)}\,\vb{A}\sqrt{\vb{B}(t)}$ reads:
\begin{equation}
    F_i(\lambda_i,t) = \lambda_i \sum_{n=0}^{+\infty}  \, \frac{\mathrm{d}\zeta_{n}^{\mathbf{B}(t)}}{\mathrm{d}t}  \lim_{N\to \infty}\frac{1}{N^{n}}\sum_{\substack{e_1, \dots, e_{n} \neq i \\ p.d.}} \prod_{j=1}^{n} \frac{\lambda_{e_j}}{ (\lambda_i - \lambda_{e_j})},
\end{equation}
where the $\zeta_{n}^{\mathbf{B}(t)}$ are defined through the series expansion of the S-transform: $\log S^{\vb{B}(t)}(z) = -\sum_n \zeta_{n}^{\mathbf{B}(t)} z^n$. In the case where $\vb{B}(t)$ is a Wishart matrix of parameter $t$, such that $\vb{B}(0)= \mathbb{I}$, one has $S^{\vb{B}(t)}(z)=(1+tz)^{-1}$ (see e.g. \cite{potters2020first}, chapter 15) and hence 
\begin{equation}
\frac{\mathrm{d}\zeta_n^{\mathbf{B}(t)}}{\mathrm{d}t} =
\begin{cases} 
(-t)^{n-1} & \text{if } n \geq 1, \\
0 & \text{if } n = 0.
\end{cases}
\end{equation}

In conclusion, we have proposed a direct interpretation of the otherwise abstract additivity of R-transforms for the sum of free matrices, in terms of the evolution of ``particles'' (the eigenvalues) interacting through two- and higher-body forces. For infinitely divisible matrices, we obtain  a non-trivial generalisation of the classic Dyson Brownian motion with Coulomb interaction. The detailed study of such a dynamical system is certainly interesting on its own, in particular when such interactions lead to the ``expulsion'' of one eigenvalue outside of the bulk, see Eq. \eqref{eq:BBP} and \cite{baik2005phase, benaych2011eigenvalues, potters2020first}. In view of the huge present activity around the Dyson Brownian motion (see e.g. \cite{dandekar2024current} and refs. therein), such a generalization would certainly be fruitful. We hope to report on this issue in the near future. 

\section*{Acknowledgements} We thank Marc Potters and Alice Guionnet for many inspiring discussions on these topics. We also thank the referee for pushing us to clarify and extend our initial results. This research was conducted within the Econophysics
\& Complex Systems Research Chair, under the aegis of
the Fondation du Risque, the Fondation de l’Ecole polytechnique, the Ecole polytechnique and Capital Fund
Management.

\clearpage
\onecolumngrid

\section{Supplementary Information} 

\subsection{SI-1: Generalization of the Burgers equation}\label{proofburgers}

Let us verify that Eq.~\eqref{burgers} follows directly from the subordination relation:
\begin{equation}\label{subordination}
    g = \g_{\vb{A}}(z - R_{\vb{B}}(g, t)),
\end{equation}
where $ \g_{\vb{A}} $ is the Stieltjes transform of $ \vb{A} $, and $g$ is a shorthand for the notation $ \g_{\vb{M}}(z, t) $. By differentiating Eq.~\eqref{subordination} with respect to $ t $ and $ z $ respectively, we obtain:

\begin{equation}\label{first}
    \partial_t g = - \left[\partial_t R_{\vb{B}} + \partial_t g \partial_g R_{\vb{B}}\right]\partial_z \g_{\vb{A}},
\end{equation}
and

\begin{equation}
    \partial_z g = (1 - \partial_z g \partial_g R_{\vb{B}}) \partial_z \g_{\vb{A}}.
\end{equation}
Thus, by multiplying Eq.~\eqref{first} by $ (1 - \partial_z g \partial_g R_{\vb{B}}) $, we have:

\begin{equation}
    (1 - \partial_z g \partial_g R_{\vb{B}}) \partial_t g = - \left[\partial_t R_{\vb{B}} + \partial_t g \partial_g R_{\vb{B}}\right] \partial_z g,
\end{equation}
which is:

\begin{equation}
    \partial_t g = - \partial_t R_{\vb{B}}(g, t) \partial_z g.
\end{equation}
This concludes the proof of Eq.~\eqref{burgers}. In fact, Eq.~\eqref{subordination} is simply the solution of the previous equation using the so-called ``methods of characteristics''.

\subsection{SI-2: Proof of equation \eqref{eq1}}\label{s2}
The key argument leading to Eq.~\eqref{eq1} is the following lemma.

\begin{lemma}\label{lemma1}
    Let $\lambda_1, \ldots, \lambda_N$ be distinct real numbers, $n\in \mathbb{N}$ and $z\in \mathbb{C}$ non-real. Then, 

    \begin{equation}
    \sum_{\underset{\text{p.d.}}{1\leq e_0, e_1, \ldots, e_n\leq N}} \frac{1}{(z- \lambda_{e_0})^2} \left[\frac{1}{\prod_{k=1}^n (\lambda_{e_0} - \lambda_{e_k})} - \frac{1}{\prod_{k=1}^n (z - \lambda_{e_k})}\right] = 0.
    \end{equation}
\end{lemma}

\begin{proof}
 Let's introduce the function

\begin{equation}
    f(\xi) = \sum_{e_0, e_1, \ldots, e_n} \frac{1}{(z- \xi)^2} \frac{1}{\prod_{k=0}^n (\xi - \lambda_{e_k})}.
\end{equation}
Since

\begin{equation}
    f(\xi) \underset{\xi\to +\infty}{\sim} \frac{1}{\xi^2}, 
\end{equation}
we know that the integral of $f$ over a contour of radius tending towards infinity should be zero. Hence, the sum of the residues of this function must be zero. Fix $1\leq m\leq N$. $f$ has a residue at $\lambda_m$ which is simply given by 

\begin{equation}
    (n+1)\sum_{e_1, \ldots, e_n\neq m} \frac{1}{(z - \lambda_m)^2} \frac{1}{\prod_{k=1}^n (\lambda_m - \lambda_{e_k})}
\end{equation}
where we make use of the symmetries of the $e_i$.
There is also a pole at $z$. Let's denote 

\begin{equation}
    P(\xi) = \prod_{k=1}^n (\xi - \lambda_{e_k}).
\end{equation}
The residue of $f$ at $z$ is simply given by 

\begin{equation}
    \sum_{e_0, e_1, \ldots, e_n} \frac{d}{d\xi}\big|_{z} \frac{1}{P(\xi)} = - \sum_{e_0, e_1, \ldots, e_n} \frac{P'(z)}{P^2(z)} = -\sum_{e_0, e_1, \ldots, e_n} \sum_{k=0}^n \frac{1}{z - \lambda_{e_k}} \frac{1}{\prod_{k=0}^n (z - \lambda_{e_k})}
\end{equation}
which gives by symmetry of $e_k$: 

\begin{equation}
    -(n+1) \sum_{e_0, \ldots, e_n} \frac{1}{(z - \lambda_{e_0})^2} \frac{1}{\prod_{k=1}^n (z - \lambda_{e_k})}.
\end{equation}
The residue theorem then states that

\begin{equation}
    \sum_m (n+1)\sum_{e_1, \ldots, e_n\neq m} \frac{1}{(z - \lambda_m)^2} \frac{1}{\prod_{k=1}^n (\lambda_m - \lambda_{e_k})} - (n+1) \sum_{e_0, \ldots, e_n} \frac{1}{(z - \lambda_{e_0})^2} \frac{1}{\prod_{k=1}^n (z - \lambda_{e_k})} = 0
\end{equation}
which proves the desired formula.

\end{proof}

 \subsection{SI-3: Perturbation Theory}

We consider a classical perturbative setting in which a deterministic matrix $ \vb{A} $ is perturbed by a matrix $ \epsilon \vb{B} $, where $ \epsilon \in \mathbb{R} $ is a small parameter. We are interested in the asymptotic expansion of the eigenvalues $\lambda_1^{\vb{M}}<...<\lambda_N^{\vb{M}}$ of the perturbed matrix $\vb{M} := \vb{A} + \epsilon \vb{B}$ in the limit $ \epsilon \to 0$. More precisely, we seek to express the eigenvalues of $ \vb{M} $ as a power series in $ \epsilon $. For each $ i $, we write:
\begin{equation}\label{defmu}
    \lambda_i^{\vb{M}} \underset{\epsilon \to 0}{=} \lambda_i^{\vb{A}} + \epsilon \mu_i^{(1)} + \epsilon^2 \mu_i^{(2)} + \cdots + \epsilon^k \mu_i^{(k)} + o(\epsilon^k),
\end{equation}
where $ \lambda_i^{\vb{A}} $ denotes the $ i $-th eigenvalue of $ \vb{A} $, and the coefficients $ \mu_i^{(j)} $ correspond to the $ j $-th order correction terms.

To compute these coefficients, we use the theory of resolvents. Recall that if $\vb{E}$ denotes a matrix, the resolvent of $\vb{E}$, denoted $\vb{G}_{\vb{E}}(z)$, is defined as the matrix $(z - \vb{E})^{-1}$, where $z$ is a complex number outside the spectrum of $\vb{E}$. Let us fix $\lambda \in \mathbb{C}$. The Cauchy formula yields
\begin{equation}
    \lambda = \frac{1}{2i\pi} \oint_\gamma \frac{z}{z - \lambda} \, dz,
\end{equation}
where $\gamma$ is a closed loop around $\lambda$. Consider $\vb{M} = \vb{A} + \epsilon\vb{B}$. The formula above implies that
\begin{equation}\label{int}
    \lambda = \frac{1}{2i\pi} \oint_\gamma z \tr(\vb{G}_{\vb{M}}(z)) \, dz,
\end{equation}
where $\lambda$ is an eigenvalue of $\vb{M}$, and $\gamma$ is a closed loop containing only $\lambda$ as an eigenvalue of $\vb{M}$. Now, let us expand $\vb{G}_{\vb{M}}(z)$ for small $\epsilon$, which gives:
\begin{equation}\label{expansion}
    \vb{G}_{\vb{M}}(z) = \vb{G}_{\vb{A}}(z) + \epsilon \vb{G}_{\vb{A}}(z) \vb{B} \vb{G}_{\vb{A}}(z) + \epsilon^2 \vb{G}_{\vb{A}}(z) \vb{B} \vb{G}_{\vb{A}}(z) \vb{B} \vb{G}_{\vb{A}}(z) + \dots.
\end{equation}
By combining Eqs. \eqref{int} and the trace of Eq.~\eqref{expansion}, we deduce that
\begin{equation}
    \mu_i^{(j)} = \sum_{1 \leq e_1, \dots, e_j \leq N} \frac{1}{2\pi i} \oint z \frac{\bra{e_1} \vb{B} \ket{e_2} \bra{e_2} \vb{B} \ket{e_3} \dots \bra{e_j} \vb{B} \ket{e_1}}{(z - \lambda_{e_1}^{\vb{A}})^2 (z - \lambda_{e_2}^{\vb{A}}) \dots (z - \lambda_{e_j}^{\vb{A}})} \, dz.
\end{equation}where the $\mu_i^{(j)}$ were defined in Eq.~\eqref{defmu}, and where the $\ket{e}$ denote the unit eigenvectors of $\vb{A}$ associated with their respective eigenvalues $\lambda_e^{\vb{A}}$. The residue theorem allows us to assert that the $\mu_i^{(j+1)}$ only include sums of fractional terms where the numerator is the product of $j$ matrix elements of $\vb{B}$, and the denominator is of the form $(\lambda_i^{\vb{A}} - \lambda_{k_1}^{\vb{A}}) \dots (\lambda_i^{\vb{A}} - \lambda_{k_j}^{\vb{A}})$, where $k_1, \dots, k_j$ are distinct from $i$. The contribution from the case where the $\lambda_{k_a}^{\vb{A}}$ are distinct is precisely of the form given by Eq.~\eqref{firsteq}.

\subsection{SI-4: Correspondence between perturbation theory and the dynamical system \eqref{eqmain}}\label{3s}

We recall the functions:
\begin{equation}\label{df2}
    \lambda_i(t) = \lambda_i^{\vb{A}} + t \mu_i^{(1)} + t^2 \mu_i^{(2)} + t^{3} \mu_i^{(3)} + t^4 \mu_i^{(4)} + o(t^4),
\end{equation}
where $\mu_i^{(1)}, \mu_i^{(2)}, \mu_i^{(3)}$ are defined in equations \eqref{per1} and \eqref{per2}, and the fourth term of the perturbation theory $\mu_i^{(4)}$ is given by:

\begin{align}\label{mu4}
\mu_i^{(4)} = & \sum_{k_1, k_2, k_3 \neq i} \frac{\bra{i} \vb{B} \ket{k_3} \bra{k_3} \vb{B} \ket{k_2} \bra{k_2} \vb{B} \ket{k_1} \bra{k_1} \vb{B} \ket{i}}{(\lambda_i - \lambda_{k_1})(\lambda_i - \lambda_{k_2})(\lambda_i - \lambda_{k_3})} \nonumber \\
& - \sum_{k_1, k_2 \neq i} \frac{\bra{i} \vb{B} \ket{k_1}^2 \bra{i} \vb{B} \ket{k_2}^2}{(\lambda_i - \lambda_{k_1})(\lambda_i - \lambda_{k_2})^2} \nonumber \\
& - \bra{i} \vb{B} \ket{i} \sum_{k_1, k_2 \neq i} \frac{\bra{i} \vb{B} \ket{k_1} \bra{k_1} \vb{B} \ket{k_2} \bra{k_2} \vb{B} \ket{i}}{(\lambda_i - \lambda_{k_1})^2 (\lambda_i - \lambda_{k_2})} \nonumber \\
& - \bra{i} \vb{B} \ket{i} \sum_{k_1, k_2 \neq i} \frac{\bra{i} \vb{B} \ket{k_1} \bra{k_1} \vb{B} \ket{k_2} \bra{k_2} \vb{B} \ket{i}}{(\lambda_i - \lambda_{k_1})(\lambda_i - \lambda_{k_2})^2} \nonumber \\
& + \bra{i} \vb{B} \ket{i}^2 \sum_{k \neq i} \frac{\bra{i} \vb{B} \ket{k}^2}{(\lambda_i - \lambda_k)^3}.
\end{align}
Here, we aim to verify that the $\lambda_i(t)$ are indeed the solutions of the dynamical system described by Eq.~\eqref{eqmain}, up to order $t^4$. For small $t$, the system \eqref{eqmain} can be approximated by

\begin{align}\label{firstapprox}
    \frac{\mathrm{d} \lambda_i(t)}{dt} &\underset{t\to 0}{=} \frac{2\kappa_2(\vb{B})t}{N} \sum_{k \neq i} \frac{1}{\lambda_i(t) - \lambda_k(t)} \nonumber \\
    &+ \frac{3\kappa_3(\vb{B}) t^{2}}{N^2} \sum_{\substack{k_1, k_2 \neq i \\ \text{p.d.}}}
    \frac{1}{(\lambda_i(t) - \lambda_{k_1}(t))(\lambda_i(t) - \lambda_{k_2}(t))} \nonumber \\
    &+ \frac{4\kappa_4(\vb{B}) t^3}{N^3} \sum_{\substack{k_1, k_2, k_3 \neq i \\ \text{p.d.}}}
    \frac{1}{(\lambda_i(t) - \lambda_{k_1}(t))(\lambda_i(t) - \lambda_{k_2}(t))(\lambda_i(t) - \lambda_{k_3}(t))} + o(t^3).
\end{align}
or alternatively, keeping all terms to order $t^3$
\begin{align}\label{approx}
    \frac{\mathrm{d} \lambda_i(t)}{dt} &\underset{t\to 0}{=} \frac{2\kappa_2(\vb{B})t}{N} \sum_{k \neq i} \frac{1}{\lambda_i^{\vb{A}} - \lambda_k^{\vb{A}}} 
    + \frac{\kappa_2(\vb{B})t^3}{N} \sum_{1 \leq k \neq i \leq N} \frac{\lambda_k''(0) - \lambda_i''(0)}{(\lambda_i^{\vb{A}} - \lambda_k^{\vb{A}})^2} \nonumber \\
    &+ \frac{3\kappa_3(\vb{B}) t^2}{N^2} \sum_{\substack{k_1, k_2 \neq i \\ \text{p.d.}}}
    \frac{1}{(\lambda_i^{\vb{A}} - \lambda_{k_1}^{\vb{A}})(\lambda_i^{\vb{A}} - \lambda_{k_2}^{\vb{A}})} \nonumber \\
    &+ \frac{4\kappa_4(\vb{B})t^3}{N^3} \sum_{\substack{k_1, k_2, k_3 \neq i \\ \text{p.d.}}}
    \frac{1}{(\lambda_i^{\vb{A}} - \lambda_{k_1}^{\vb{A}})(\lambda_i^{\vb{A}} - \lambda_{k_2}^{\vb{A}})(\lambda_i^{\vb{A}} - \lambda_{k_3}^{\vb{A}})} + o(t^3).
\end{align}
where $\lambda''(0)$ is a short-hand for $\mathrm{d}^2\lambda/\mathrm{d}t^2|_{t=0}$. Now,
\begin{align}
    \frac{\kappa_2(\vb{B})}{N}\sum_{k \neq i} \frac{\lambda_k''(0) - \lambda_i''(0)}{(\lambda_i^{\vb{A}} - \lambda_k^{\vb{A}})^2} &= \frac{2\kappa_2(\vb{B})^2}{N^2}\sum_{1 \leq k \neq i \leq N} \frac{1}{(\lambda_i^{\vb{A}} - \lambda_k^{\vb{A}})^2} \left(\sum_{j\neq k} \frac{1}{\lambda_k^{\vb{A}} - \lambda_j^{\vb{A}}} - \sum_{j\neq i} \frac{1}{\lambda_i^{\vb{A}} - \lambda_j^{\vb{A}}}\right).
\end{align}
Since
\begin{equation}
    \sum_{j\neq k} \frac{1}{\lambda_k^{\vb{A}} - \lambda_j^{\vb{A}}} - \sum_{j\neq i} \frac{1}{\lambda_i^{\vb{A}} - \lambda_j^{\vb{A}}} = \sum_{j\neq k, i} \frac{\lambda_i^{\vb{A}} - \lambda_k^{\vb{A}}}{(\lambda_k^{\vb{A}} - \lambda_j^{\vb{A}})(\lambda_i^{\vb{A}} - \lambda_j^{\vb{A}})} + \frac{2}{\lambda_k^{\vb{A}} - \lambda_i^{\vb{A}}},
\end{equation}
then
\begin{equation}
    \frac{\kappa_2(\vb{B})}{N}\sum_{k \neq i} \frac{\lambda_k'(0) - \lambda_i'(0)}{(\lambda_i^{\vb{A}} - \lambda_k^{\vb{A}})^2} = \\
    \frac{2\kappa_2(\vb{B})^2}{N^2} \left[ \sum_{\substack{k, j \neq i \\ k \neq j}}
 \frac{1}{(\lambda_k^{\vb{A}} - \lambda_j^{\vb{A}})(\lambda_i^{\vb{A}} - \lambda_k^{\vb{A}})(\lambda_i^{\vb{A}} - \lambda_j^{\vb{A}})} + \sum_{k\neq i}\frac{2}{(\lambda_k^{\vb{A}} - \lambda_i^{\vb{A}})^3}\right].
\end{equation}
The first sum on the right-hand side vanishes by using symmetry between the indices $k$ and $j$. Finally, by returning to Eq.~\eqref{approx}, we obtain
\begin{align}\label{finalapprox}
    \frac{\mathrm{d} \lambda_i(t)}{dt} &\underset{t\to 0}{=} \frac{2\kappa_2(\vb{B})t}{N} \sum_{k \neq i} \frac{1}{\lambda_i^{\vb{A}} - \lambda_k^{\vb{A}}} \nonumber \\
    &+ \frac{3\kappa_3(\vb{B}) t^2}{N^2} \sum_{\substack{k_1, k_2 \neq i \\ \text{p.d.}}} \frac{1}{(\lambda_i^{\vb{A}} - \lambda_{k_1}^{\vb{A}})(\lambda_i^{\vb{A}} - \lambda_{k_2}^{\vb{A}})} \nonumber \\
    &+ 4t^3\Bigg(\frac{\kappa_4(\vb{B})}{N^3} \sum_{\substack{k_1, k_2, k_3 \neq i \\ \text{p.d.}}} \frac{1}{(\lambda_i^{\vb{A}} - \lambda_{k_1}^{\vb{A}})(\lambda_i^{\vb{A}} - \lambda_{k_2}^{\vb{A}})(\lambda_i^{\vb{A}} - \lambda_{k_3}^{\vb{A}})} \nonumber \\
    &\quad - \frac{\kappa_2(\vb{B})^2}{N^2} \sum_{k \neq i} \frac{1}{(\lambda_i^{\vb{A}} - \lambda_k^{\vb{A}})^3} \Bigg) + o(t^3).
\end{align}
In the $N \to \infty$ limit, we expect the functions defined in Eq.~\eqref{df2} to be approximated by
\begin{equation}\label{def}
    \lambda_i(t) = \lambda_i^{\mathbf{A}} + t \mathbb{E}[\mu_i^{(1)}] + t^2 \mathbb{E}[\mu_i^{(2)}] + t^{3} \mathbb{E}[\mu_i^{(3)}] + t^4 \mathbb{E}[\mu_i^{(4)}],
\end{equation}
where $\mathbb{E}[\mu_i^{(1)}] = 0$, and $\mathbb{E}[\mu_i^{(2)}]$, $\mathbb{E}[\mu_i^{(3)}]$ are given in Eq.~\eqref{expected}. Additionally, the formula ~\eqref{cyclicformula} yields
\begin{equation}
\mathbb{E}\left[\mu_i^{(4)}\right] = \frac{\kappa_4(\mathbf{B})}{N^3} \sum_{\substack{k_1, k_2, k_3 \neq i \\ \text{p.d.}}}
 \frac{1}{(\lambda_i^{\vb{A}} - \lambda_{k_1}^{\vb{A}})(\lambda_i^{\vb{A}} - \lambda_{k_2}^{\vb{A}})(\lambda_i^{\vb{A}} - \lambda_{k_3}^{\vb{A}})} \\
- \frac{\kappa_2(\mathbf{B})^2}{N^2} \sum_{k \neq i} \frac{1}{(\lambda_i^{\vb{A}} - \lambda_k^{\vb{A}})^3}.
\end{equation}
We then verify that the functions $\lambda_i$ in Eq.~\eqref{def} are solutions of the approximate system above as well as the system in Eq.~\eqref{eqmain}, up to order $t^4$.
\subsection{SI-5: A lemma}

\begin{lemma}\label{lemma}
Let $(\eta_t)$ be IID random variables drawn from some distribution $\psi(\eta)$, and $(\vec{U}_t)$ be random vectors chosen uniformly on the unit sphere of $\mathbb{R}^N$, with $\|\vec{U}_t\|^2=1$. Define $T := \lfloor N \Delta t \rfloor$, and set:

\begin{equation}
    \Delta \mathbf{B} = \sum_{t=1}^T \eta_t \left|\vec{U}_t\right\rangle \left\langle \vec{U}_t\right|.
\end{equation}
Let $j \geq 2$ and consider distinct, orthonormal vectors $b_1, \dots, b_j \in \mathbb{R}^N$. Then, there exists a constant $c_j > 0$ such that:

\begin{equation}
    \mathbb{V}\left[\langle b_1 | \vb{B} | b_2 \rangle \langle b_2 | \Delta\vb{B}| b_3 \rangle \dots \langle b_{j-1} | \Delta\vb{B} | b_j \rangle\right] \underset{N \to +\infty}{\sim} \frac{c_j (\Delta t)^j}{N^j}.
\end{equation}
\end{lemma}

\begin{proof}
Using formula \eqref{cyclicformula} or expanding the terms of $\Delta \mathbf{B}$ explicitly, one shows that:

\begin{equation}\label{justeexp}
    \mathbb{E}\left[\langle b_1 |\Delta \vb{B}| b_2 \rangle \langle b_2 |\Delta \vb{B}| b_3 \rangle \dots \langle b_{j-1} |\Delta \vb{B}| b_j \rangle\right] \underset{N \to +\infty}{\sim} \frac{\mathbb{E}[\eta^j] \Delta t}{N^{j-1}}.
\end{equation}
One can expand:

\begin{equation}\label{developpement}
\left[\langle b_1 |\Delta \vb{B}| b_2 \rangle \dots \langle b_{j-1} |\Delta \vb{B}| b_j \rangle\right]^2 = \sum_{1\leq t_1,t_1',\dots,t_j,t_j'\leq T} \prod_{k=1}^j \left[\eta_{t_k}\eta_{t_k'}[\vec{U}_{t_{k-1}}]_{b_k}[\vec{U}_{t_k}]_{b_k}[\vec{U}_{t_{k-1}'}]_{b_k}[\vec{U}_{t_k'}]_{b_k}\right],
\end{equation}with the convention $t_0 = t_j$, $t_0' = t_j'$, and $[\vec{U}_t]_{b_k} = \langle \vec{U}_t | b_k \rangle$. To evaluate the expectation, recall that a random unit vector $\vec{U}_t$ uniformly distributed on the unit sphere can be generated as $\vec{U}_t = X_t / |X_t|$, where $X_t$ is a standard Gaussian vector in $\mathbb{R}^N$. In this representation, the coordinates $[\vec{U}_t]_{b_k}$ are ratios of Gaussian components to the Euclidean norm of the full vector. Using this representation, the dominant contribution to the expectation of \eqref{developpement} as $N \to +\infty$ thus scales as $
c_j T^j/N^{2j}$ where $c_j$ is a fixed strictly positive constant. Hence, combined with \eqref{justeexp}, we find:

\begin{equation}
\mathbb{V}\left[\langle b_1 |\Delta \vb{B}| b_2 \rangle \dots \langle b_{j-1} |\Delta \vb{B}| b_j \rangle\right] \underset{N\to \infty}{\sim} \frac{c_j (\Delta t)^j}{N^j} - \frac{\mathbb{E}[\eta^j]^2 (\Delta t)^2}{N^{2(j-1)}} \underset{N\to \infty}{\sim} \frac{c_j (\Delta t)^j}{N^j}
\end{equation}for $j \geq 3$.

For the special case $j = 2$, we start from

\begin{equation}
\left[\langle b_1|\Delta \mathbf{B}|b_2\rangle\langle b_2|\Delta \vb{B}|b_1\rangle\right]^2 = \sum_{1\leq t_1,t_1',t_2,t_2'\leq T} \eta_{t_1}\eta_{t_2}\eta_{t_1'}\eta_{t_2'}[\vec{U}_{t_1}]_{b_1}[\vec{U}_{t_2}]_{b_1}[\vec{U}_{t_1'}]_{b_1}[\vec{U}_{t_2'}]_{b_1}[\vec{U}_{t_1}]_{b_2}[\vec{U}_{t_2}]_{b_2}[\vec{U}_{t_1'}]_{b_2}[\vec{U}_{t_2'}]_{b_2}.
\end{equation}
Taking the expectation as $N \to +\infty$, the leading term is

\begin{equation}
\frac{3\mathbb{E}[\eta^2]^2T^2}{N^4} = \frac{3\mathbb{E}[\eta^2]^2(\Delta t)^2}{N^2}.
\end{equation}
Since
\begin{equation}
\mathbb{E}\left[\langle b_1|\Delta \vb{B}|b_2\rangle\langle b_2|\Delta \vb{B}|b_1\rangle\right] = \frac{\mathbb{E}[\eta^2]\Delta t}{N},
\end{equation}we can write
\begin{equation}
\mathbb{V}\left[\langle b_1|\Delta \vb{B}|b_2\rangle\langle b_2|\mathbf{B}|b_1\rangle\right] \underset{N \to +\infty}{\sim} \frac{2\mathbb{E}[\eta^2]^2(\Delta t)^2}{N^2},
\end{equation}
which concludes the proof.
\end{proof}

\newpage
\bibliographystyle{apsrev4-2}
\bibliography{References.bib}

\end{document}